\documentclass[a4paper]{scrartcl}

\pdfoutput=1

\usepackage{authblk}
\usepackage{graphicx}
\usepackage{latexsym}
\usepackage{textcomp}
\usepackage{stmaryrd}
\usepackage[nointegrals]{wasysym}
\usepackage{wrapfig}
\usepackage{listings}
\usepackage{color}
\usepackage{footnote}
\usepackage{multirow} 
\usepackage{array}
\usepackage{xspace}
\usepackage{units}
\usepackage{enumitem}
\usepackage{subfig}
\usepackage[english]{babel}
\usepackage{amsmath}
\usepackage{amssymb}
\usepackage{alltt}
\usepackage{algorithmic}
\usepackage{listings}
\usepackage{tikz}
\usetikzlibrary{arrows,shapes,snakes,automata,backgrounds,petri,patterns,spy}
\usepackage{placeins}
\usepackage{calculator}
\usepackage{tabularx}
\usepackage{url}

\usepackage{amsthm}
\newtheorem{definition}{Definition}
\newtheorem{lemma}{Lemma}
\newtheorem{theorem}{Theorem}
\newtheorem{corollary}{Corollary}

\newcolumntype{Y}{>{\hfill\arraybackslash}X}

\newcommand{\comment}[1]{ }
\newcommand{\vect}[1]{\ensuremath{\mathbf{#1}}}

\newcommand{\Pre}{\vect{Pre}} 
\newcommand{\Post}{\vect{Post}}

\newcommand{\TA}[0]{{TA}\xspace}
\newcommand{\TPN}[0]{{TPN}\xspace}
\newcommand{\DTPN}[0]{{DTPN}\xspace}

\newcommand{\SCG}[0]{\ensuremath{SCG}\xspace}

\newcommand{\low}[1]{{\downarrow}#1}
\newcommand{\high}[1]{{\uparrow}#1}

\newcommand{\Nat}{\ensuremath{\mathbb{N}}\xspace}

\newcommand{\Rat}{\ensuremath{\mathbb{Q}}\xspace}

\newcommand{\pReal}{\ensuremath{\mathbb{R}_{\ge 0}\xspace}}
\newcommand{\Itrv}{\ensuremath{\mathbb{I}\xspace}}

\makeatletter
\newcommand{\dotminus}{\mathbin{\text{\@dotminus}}}
\newcommand{\@dotminus}{%
  \ooalign{\hidewidth\raise1ex\hbox{.}\hidewidth\cr$\m@th-$\cr}%
}
\makeatother

\newcommand{\trans}[2][{}]{\mathbin{\smash{\overset{{#2}}{{\to}}_{#1}}}}
\newcommand{\wtrans}[2][{}]{\mathbin{\smash{\overset{{#2}}{{\Rightarrow}}_{#1}}}}
\newcommand{\interp}[1]{{[\![}\,{#1}\,{]\!]}}
\newcommand{\RR}[0]{{\cal R}}

\newcommand{\tI}[1][]{\ensuremath{\varphi_{#1}}}

\newcommand{\sI}[1]{\ensuremath{x_{#1}}}
\newcommand{\ssI}[1]{\ensuremath{y_{#1}}}
\newcommand{\falpha}{\ensuremath{A}}
\newcommand{\fbeta}{\ensuremath{B}}
\newcommand{\falphas}{\ensuremath{A^s}}
\newcommand{\fbetas}{\ensuremath{B^s}}

\newcommand{\E}{{\cal E}}

\newcommand{\SIF}{\ensuremath{\textbf{I}_s}\xspace}
\newcommand{\DIF}{\ensuremath{\textbf{I}_d}\xspace}
\newcommand{\half}{\nicefrac{1}{2}}

\DeclareMathOperator{\pers}{prs} 
\DeclareMathOperator{\nenabl}{nbl}
\DeclareMathOperator{\cfl}{\Join}
\DeclareMathOperator{\NBL}{NBL}
\DeclareMathOperator{\CFL}{CFL}

\newcommand{\ndname}[1]{\textsf{#1}\xspace}

\newcommand{\Marks}{\ensuremath{{\cal M}}}
\newcommand{\IM}{\Itrv_{\Marks}}
\newcommand{\divides}{~|~}

\newcommand{\pp}[2]{p_{{#1},{#2}}}
\newcommand{\pq}[2]{q_{{#1},{#2}}}

\begin{document}


\title{Time Petri Nets with Dynamic Firing\\
  Dates: Semantics and Applications}

\author[1,2]{Silvano Dal Zilio}
\author[1,2]{{\L}ukasz Fronc}
\author[1,2]{Bernard Berthomieu}
\author[1,3]{Fran{\c{c}}ois Vernadat}
\affil[1]{CNRS, LAAS, F-31400 Toulouse,
  France}
\affil[2]{Univ de Toulouse, LAAS, F-31400 Toulouse, France}
\affil[3]{Univ de Toulouse, INSA, LAAS, F-31400 Toulouse, France}
\date{}
\maketitle
\begin{abstract}
  We define an extension of time Petri nets such that the time at
  which a transition can fire, also called its \emph{firing date}, may
  be dynamically updated. Our extension provides two mechanisms for
  updating the timing constraints of a net. First, we propose to
  change the static time interval of a transition each time it is
  newly enabled; in this case the new time interval is given as a
  function of the current marking. Next, we allow to update the firing
  date of a transition when it is persistent, that is when a
  concurrent transition fires. We show how to carry the widely used
  {state class} abstraction to this new kind of time Petri nets and
  define a class of nets for which the abstraction is exact. We show
  the usefulness of our approach with two applications: first for
  scheduling preemptive task, as a poor man's substitute for
  stopwatch, then to model hybrid systems with non trivial continuous
  behavior.
\end{abstract}

\section{Introduction}
\label{sec:introduction}

A {\em Time Petri Net}~\cite{Merlin74astudy,BPV06} ({\TPN}) is a Petri
net where every transition is associated to a static time interval
that restricts the date at which a transition can fire. In this model,
time progresses with a common rate in all the transitions that are
enabled; then a transition $t$ can fire if it has been continuously
enabled for a time $\theta_t$ and if the value of $\theta_t$ is in the
static time interval, denoted $\SIF(t)$. The term {static time
  interval} is appropriate in this context. Indeed, the constraint is
immutable and do not change during the evolution of the net. In this
paper, we lift this simple restriction and go one step further by also
updating the timing constraint of persistent transitions, that is
transitions that remain enabled while a concurrent transition fires.
In a nutshell, we define an extension of \TPN where the time at which
a transition can fire, also called its \emph{firing date}, may be
dynamically updated. We say that these transitions are fickle and we
use the term Dynamic \TPN to refer to our extension.

Our extension provides two mechanisms for updating the timing
constraints of a net. First, we propose to change the static time
interval of a transition each time it is newly enabled. In this case
the new time interval $\SIF(t,m)$ is obtained as a function of the
current marking $m$ of the net. Likewise, we allow to update the
deadline of persistent transitions using an expression of the form
$\DIF(t,m,\tI[t])$, that is based on the previous firing date of
$t$. The first mechanism is straightforward and quite similar to an
intrinsic capability of Timed Automata (\TA); namely the possibility
to compare a given clock to different constants depending on the
current state. Surprisingly, it appears that this extension has never
been considered in the context of \TPN. The second mechanism is far
more original. To the best of our knowledge, it has not been studied
before in the context of \TPN or \TA, but there are some similarities
with the updatable timed automata of Bouyer et
al.~\cite{Bouyer2004291}.

The particularity of timed models, such as \TPN, is
that state spaces are typically infinite, with finite representations
obtained by some abstractions of time. In the case of \TPN, states are
frequently represented using composite abstract states, or \emph{state
  classes}, that captures a discrete information (e.g. the marking)
together with a timing information (represented by systems of
difference constraints or zones). We show how to carry the {state
  class} abstraction to our extended model of \TPN. We only
obtain an over-approximation of the state space in the most general
case, but we define a class of nets for which the abstraction
is exact. We conjecture that our approach could be used in other
formal models for real-time systems, such as timed automata for
instance.

There exist several tools for reachability analysis of \TPN based on
the notion of {state class graph}~\cite{BM82,BD91}, like for example
Tina~\cite{BRV04} or Romeo~\cite{gardey2005romeo}. Our construction
provides a simple method for supporting fickle transitions in these
tools. Actually, our extension has been implemented inside the tool
Tina in a matter of a few days. We have used this extension of Tina to
test the usefulness of our approach in the context of two possible
applications: first for scheduling preemptive task, as a poor man's
substitute for stopwatch; next to model dynamical systems with non
trivial continuous behavior.

\paragraph*{Outline of the paper and contributions.} We define the
semantics of \TPN with dynamic firing dates in Sect.~\ref{sec:tpn}. We
prove that we directly subsume the class of ``standard'' \TPN and
that our extension often leads to more concise models. In
Sect.~\ref{sec:simple-example-dynam}, we motivate our extension by
showing how to implement the Quantized State System (QSS)
method~\cite{cellier06:_contin_system_simul}. This application
underlines the advantage of using an asynchronous approach when
modeling hybrid systems. Section~\ref{sec:state-class-abstr} provides
an incremental construction for the state class graph of a dynamic
\TPN. Before concluding, we give some experimental results for two
possible applications of dynamic \TPN.

\section{Time Petri nets and Fickle Transitions}\label{sec:tpn}


A {\em Time Petri net} is a Petri net where transitions are decorated
with static time intervals that constrain the time a transition can
fire. We denote $\Itrv$ the set of possible time intervals. We use a
dense time model in our definitions, meaning that we choose for
$\Itrv$ the set of real intervals with non negative rational
end-points.
To simplify the definitions, we only consider the case of closed
intervals, $[a, b]$, and infinite intervals of the form $[a,
+\infty)$. For any interval $i$ in $\Itrv$, we use the notation
$\low{i}$ for its left end-point and $\high{i}$ for its right
end-point (or $\infty$ if $i$ is unbounded).

We use the expression Dynamic \TPN (\DTPN) when it is necessary to
make the distinction between our model and more traditional
definitions of \TPN.  With our notations, a dynamic \TPN\ is a tuple
$\langle {P},{T},{\Pre},{\Post},m_0,\SIF, \DIF \rangle$ in which:
\begin{itemize}
\item $\langle {P},{T},{\Pre},{\Post}, m_0 \rangle$ is a Petri net,
  with ${P}$ the set of places, ${T}$ the set of transitions, $m_0 : P
  \rightarrow \Nat$ the initial marking, and ${\Pre},~ {\Post} : {T}
  \rightarrow {P} \rightarrow \Nat$ the precondition and
  postcondition functions.
\item $\SIF$ is the \emph{static interval function}, that associates a
  time interval (in $\Itrv$) to every transition (in $T$).
\item \DIF is the \emph{dynamic interval function}. It will be used to
  update the firing date of persistent transitions.

\end{itemize}

We slightly extend the ``traditional'' model of \TPN and allow to
define the static time interval of a transition as a function of the
markings, meaning that $\SIF$ is a function of $T \to (P \to \Nat) \to
\Itrv$. We will often used the curryied function $\SIF(t)$ to denote
the mapping from a marking $m$ to the time interval $\SIF(t,m)$. 

We also add the notion of \emph{dynamic interval function}, \DIF, that
is used to update the firing date of persistent transitions. The idea
is to update the firing date $\tI[t]$ of a persistent transition $t$
using a function of $\tI[t]$. Hence \DIF is a function of $T \to (P
\to \Nat) \to \pReal \to \Itrv$. For example, a transition $t$ such
that $\DIF(t, m, \theta) = [\theta+1, \theta+2]$, for all $\theta \geq
0$, models an event that is delayed by between $1$ and $2$ unit of
time (u.t.) when a concurrent transition fires.

\subsection{A Semantics for Time Petri Nets Based on Firing Functions}

As usual, we define a \emph{marking} $m$ of a \TPN as a function $m :
{P} \rightarrow \Nat$ from places to integers. A transition $t \in
{T}$ is {\em enabled} at $m$ if and only if $m \ge \Pre(t)$ (we use
the pointwise comparison between functions). We denote ${\cal E}(m)$
the set of transitions enabled at $m$.


A \emph{state} of a {\TPN} is a pair $s = (m, \tI)$ in which $m$ is a
marking and $\tI: {T} \rightarrow \pReal$ is a mapping, called the
\emph{firing function} of $s$, that associates a firing date to every
transition enabled at $m$. Intuitively, if $t$ is enabled at $m$, then
$\tI[t]$ is the date (in the future, from now) at which $t$ should
fire. Also, the transitions that may fire from a state $(m, \tI)$ are
exactly the transitions $t$ in $\E(m)$ such that $\tI[t]$ is minimal;
they are the first scheduled to fire.

For any date $\theta$ in $\pReal$, we denote $\tI \dotminus \theta$
the partial function that associates the transition $t$ to the value
$\tI[t] - \theta$, when $\tI[t] \geq \theta$, and that is undefined
elsewhere. This operation is useful to model the effect of time
passage on the enabled transitions of a net. We say that the firing
function $\tI \dotminus \theta$ is well-defined if it is defined on
exactly the same transitions than $\tI$.

The following definitions are quite standard. The semantics of a \TPN
is a Kripke structure $\langle S,S_0,\rightarrow\rangle$ with only two
possible kind of actions: either $s \trans{t} s'$ (meaning that the
transition $t \in T$ is fired from $s$); or $s \trans{\theta} s'$,
with $\theta \in \pReal$ (meaning that we let time $\theta$ elapse
from $s$). A transition $t$ may fire from the state $(m,\tI)$ if $t$
is enabled at $m$ and firable instantly (that is $\tI[t] = 0$). In a
state transition $(m, \tI) \trans{t} (m', \tI')$, we say that a
transition $k$ is \emph{persistent} (with $k \neq t$) if it is also
enabled in the marking $m - \Pre(t)$, that is if $m - \Pre(t) \geq
\Pre(k)$. The transitions that are enabled at $m'$ and not at $m$ are
called \emph{newly enabled}. We define the predicates $\pers$ and
$\nenabl$ that describe the set of persistent and newly enabled
transitions after $t$ fires from $m$:
\[
\begin{array}{lcl}
\pers(m,t) &=& \{ k \in \E(m) \mid  m - \Pre(t) \geq \Pre(k) \}\\
\nenabl(m,t) &=& \{ k \in (T \setminus \E(m)) \cup \{ t \} \mid  m - \Pre(t) + \Post(t) \geq \Pre(k) \}
\end{array}
\]
We use these two predicates to define the semantics of \DTPN.

\begin{definition}\label{def:tpnstate}
  The semantics of a {\DTPN} $\langle {P},{T},{\Pre},{\Post},m_0,\SIF,
  \DIF \rangle$ is the timed transition system $SG = \langle
  S,S_0,\rightarrow\rangle$ such that:
  \begin{itemize}
  \item $S$ is the set of states of the {\TPN};
  \item $S_0$, the set of initial states, is the subset of states of
    the form $(m_0, \tI)$, where $m_0$ is the initial marking and
    $\tI[t] \in \SIF(t, m_0)$ for every $t$ in $\E(m_0)$;
  \item the state transition relation ${\rightarrow} \subseteq S
    \times ({{T} \cup \pReal})\! \times S$ is the smallest relation
    such that for all state $(m, \tI)$ in $S$:
   \begin{itemize}
   \item[(i)] if $t$ is enabled at $m$ and $\tI[t] = 0$ then $(m, \tI)
     \trans{t} (m', \tI')$ where $m' = m - \Pre(t) + \Post(t)$ and
     $\tI'$ is a firing function such that $\tI[k]' \in \DIF(k, m',
     \tI[k])$ for all persistent transition $k \in \pers(m,t)$ and
     $\tI[k]' \in \SIF(k, m')$ otherwise.
   \item[(ii)] if $\tI \dotminus \theta$ is well-defined then $(m,\tI)
     \trans{\theta} (m,\tI \dotminus \theta)$.
   \end{itemize}
\end{itemize}

\label{def:sg}
\end{definition}

The state transitions labelled over $T$ (case $(i)$ above) are the
{\em discrete} transitions, those labelled over $\pReal$ (case $(ii)$)
are the {\em continuous}, or time elapsing, transitions. It is clear
from Definition~\ref{def:sg} that, in a discrete transition $(m, \tI)
\trans{t} (m', \tI')$, the transitions enabled at $m'$ are exactly
$\pers(m,t) \cup \nenabl(m,t)$. In the target state $(m', \tI')$, a
newly enabled transition $k$ get assigned a firing date picked ``at
random'' in $\SIF(k,m')$. Similarly, a persistent transition $k$ get
assigned a firing date in $\DIF(k,m', \tI[k])$.
Because there may be an infinite number of transitions, the state
spaces of {\TPN} are generally infinite, even when the net is
bounded. This is why we introduce an abstraction of the semantics in
Sect.~\ref{sec:state-class-abstr}.


We can define two simple extensions to \DTPN. First, we can use a
special treatment for re-initialized transitions; transitions that are
enabled before $t$ fires and newly-enabled after. In this case we
could use the previous firing date to compute the static
interval. Then, in the interval functions \DIF and \SIF, we can use
the ``identifier'' of the transition that fires in addition to the target
marking, $m'$. These extensions preserve the results described in this
paper

Our definitions differ significantly from the semantics of \TPN
generally used in the literature. For instance, in the works of
Berthomieu et al.~\cite{Ber03b}, states are either based on
\emph{clocks}---that is on the time elapsed since a transition was
enabled---or on \emph{firing domains} (also called time zones)---that
abstract the sets of possible ``time to fire'' using intervals. Our
choice is quite close to the \TPN semantics based on firing domains
(in particular we have the same set of traces) and is similar in
spirit to the semantics used by Vicario et
al.~\cite{DBLP:journals/tse/VicarioSC09} for reasoning about
Stochastic Time Petri nets. We made the choice of an unorthodox
semantics to simplify our definition of firing date. We conjecture
that most of our definitions can be transposed to a clock-based
semantics.

\subsection{Interesting Classes of \DTPN}
\label{sec:inter-subcl}

In the standard semantics of \TPN~\cite{Merlin74astudy}, the firing
date of a persistent transition is left unchanged. We can obtain a
similar behavior by choosing for $\DIF(t, m,\theta)$ the time interval
$[\theta, \theta]$. We say in this case that the dynamic interval
function is \emph{trivial}. Another difference with respect to the
standard definition of \TPN is the fact that the (static!) time
interval of a transition may change. We say that a dynamic net is a
\TPN if its static function, $\SIF$, is constant and its dynamic
function, $\DIF$, is trivial.  We say that a \DTPN is \emph{weak} if
only the function $\DIF$ is trivial.
We show that \TPN are as expressive than weak \DTPN when the nets are
bounded. Weak nets are still interesting though, since the use of
non-constant interval functions can lead to more concise models.  On
the other hand, the results of Sect.~\ref{sec:state-class-abstr} show
that, even in bounded nets, fickle transitions are more expressive
than weak ones.

\begin{theorem}\label{th:standard}
  For every weak \DTPN that has a finite set of reachable markings,
  there is a \TPN that has an equivalent semantics.
\end{theorem}

\begin{proof} see Appendix~\ref{sec:proof-theorem-1}.
\end{proof}

We define a third class of nets, called \emph{translation \DTPN},
obtained by restricting the dynamic interval function \DIF. This class
arises naturally during the definition of the State Class Graph
construction in Sect.~\ref{sec:state-class-abstr}. Intuitively, with
this restriction, a persistent transition can only shift its firing
date by a ``constant time''. The constant can be negative and may be a
function of the marking. More precisely, we say that a \DTPN is a
translation if, for every transitions $t$, there are two functions
$\kappa_{1}$ and $\kappa_{2}$ from $(P \to \Nat) \to \Rat$ such that
${\DIF(t,m,\theta)}$ is the time interval $[\falpha, \fbeta]$ where
$\falpha = \max(0, \theta + \kappa_{1}(m))$ and $\fbeta =
\max(\falpha, \theta + \kappa_{2}(m))$. (The use of $\max$ in the
definition of $\falpha, \fbeta$ is necessary to accomodate negative
constants $\kappa_i(m)$.)


\subsection{Interpretation of the Quantized State System Model}
\label{sec:simple-example-dynam}


With the addition of fickle transitions, it is possible to model
systems where the timing constraints of an event depend on the
current state. This kind of situations arises naturally in
practice. For instance, we can use the function \SIF to model the fact
that the duration of a communication depends on the length of a
message. Likewise, we can use the fickle function \DIF when modeling
the typical workflow of a conference, in which a deadline may be
postponed when particular events occurs.

In this section, we consider a simple method for analyzing the
behavior of a system with one continuous variable, $x$, governed by
the ordinary differential equation $\dot{x} = f(x)$. The idea is to
define a \TPN that computes the value $x(\theta)$ of the variable $x$
at the date $\theta$. To this end, we use an extension of \TPN with
shared variables, $x, y, \dots$, where every transition may be guarded
by a boolean predicate (\textbf{on}\,$b$) and such that, upon firing,
a transition can update the environment (using a sequence of
assignments, \textbf{do}\,$e$). This extension of \TPN with shared
variables can already be analyzed using the tool Tina.

\begin{figure}
  \centering
  {\includegraphics[width=0.48\textwidth]{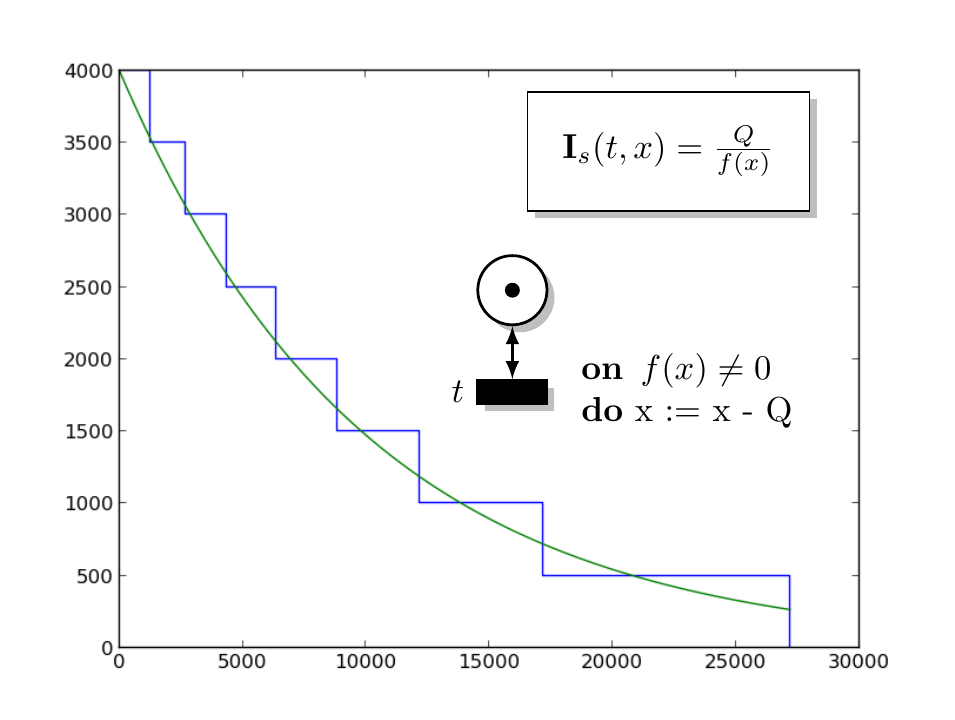}}
  \hfill
  {\includegraphics[width=0.48\textwidth]{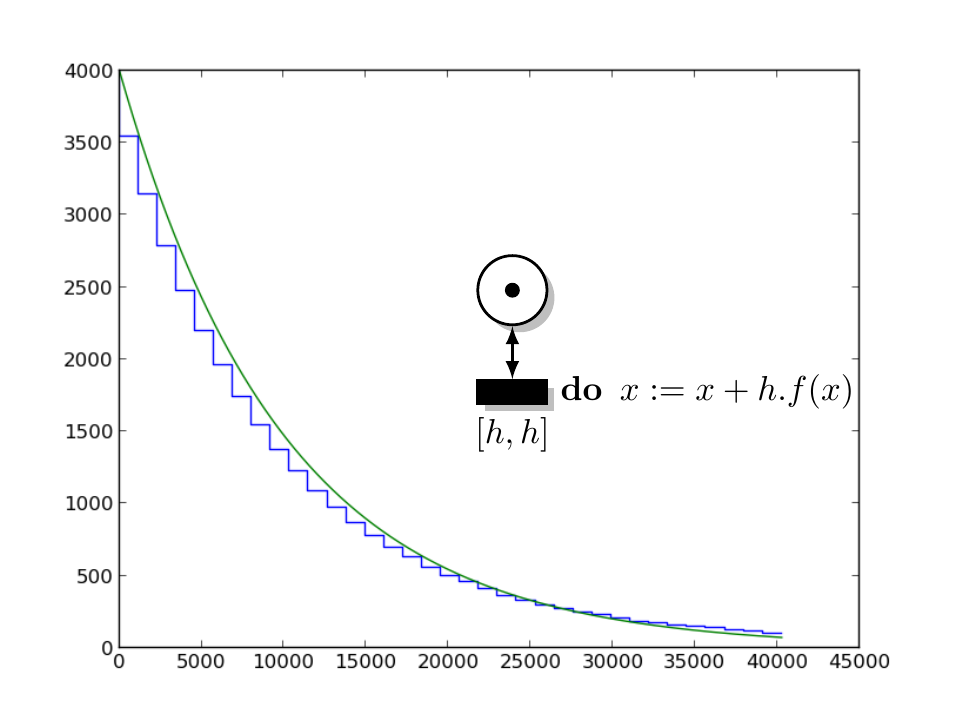}}
  \caption{A simple QSS simulation (left) and the Euler method (right)
    for $\dot{x} = -x$. ($Q = 500$, $h = 1150$, global error smaller
    than $500$.)\label{fig::qss}}
\end{figure}

The simplest solution is based on the Euler forward method. This is
modeled by the \TPN of Fig.~\ref{fig::qss} (right) that periodically
execute the instruction $x := x + h.f(x)$ every $h$ (the value of the
time step, $h$, is the only parameter of the method). This solution is
a typical example of synchronous system, where we sample the evolution
of time using a ``quantum of time''. A synchronous approach answers
the following question: given the value of $x$ at time $k.h$, what is
its value at time $(k+1).h$?

The second solution is based on the Quantized State System (QSS)
method~\cite{cellier06:_contin_system_simul,cellier2008quantized},
which can be interpreted as the dual---the asynchronous
counterpart---of the Euler method. QSS uses a ``quantum of value'',
$Q$, meaning that we only consider discrete values for $x$, of the
form $k.Q$ with $k \in \Nat$. The idea is to compute the time
necessary for $x$ to change by an amount of $Q$. To
paraphrase~\cite{cellier2008quantized}, the QSS method answers the
following modified question: given that $x$ has value $k.Q$, what is
the earliest time at which $x$ has value $(k \pm 1). Q$? This method
has a direct implementation using fickle transitions: at first
approximation, the time $\tI[t]$ for $x$ to change by an amount of $Q$
is given by the relation $(x \pm Q) = x + \tI[t] . f(x)$, that is
$\tI[t] = \nicefrac{Q}{|f(x)|}$. We have that the \emph{time slope} of
$x$ is equal to $\nicefrac{1}{f(x)}$. The role of the guard $f(x) \neq
0$ on transition $t$ is to avoid pathological values for the slope;
when $f(x)$ is nil the value of $x$ stays constant, as needed.

We compare the results obtained with these two different solutions in
Fig.~\ref{fig::qss}, where we choose $f(x) = - x$ and $x(0) =
4000$. Each plot displays the evolution of the \TPN compared to the
analytic solution, in this case $x(\theta) = 4000 e^{-\theta}$.
Numerical methods are of course approximate; in both cases (Euler and
QSS) the global error is proportional to the quantum.  The plots are
obtained with the largest quantum values giving a global error smaller
than $500$, that is a step $h$ of $1150$ and a quantum $Q$ of
$500$. The dynamic \TPN has $10$ states while the standard \TPN has
$38$. The ratio improves when we try to decrease the global error. For
instance, for an error smaller than $100$ (which gives $Q = 100$ and
$h = 250$) we have $42$ states against $182$. We observe that in this
case the ``asynchronous'' solution is more concise than the
synchronous one.

The Euler method is the simplest example in a large family of
iterative methods for approximating the solutions of differential
equations. The QSS method used in this section can be enhanced in just
the same way, leading to more precise solutions, with better numerical
stability. Some of the improved QSS methods have been implemented in
our tool, but we still experiment the effect of numerical instability
on some stiff systems. In these cases, the synchronous approach (that
is deterministic) may sometimes exhibit better performances.


Although we make no use of the fickle function \DIF here,
it arises naturally when the system has multiple variables. Consider a
system with two variables, $x, y$, such that $x = f(x,y)$. We can use
the same solution than in Fig.~\ref{fig::qss} to model the evolution
of $x$ and $y$. When the value of $x$ just changes, the next update is
scheduled at the date $\nicefrac{Q}{f(x,y)}$ (the time slope is $f_1 =
\nicefrac{1}{f(x,y)}$). If the value of $y$ is incremented before this
deadline---say that the remaining time if $\theta_1$---we need to
update the time slope and use the new value $f_2 =
\nicefrac{1}{f(x,y+Q)}$.

\begin{figure}
  \centering
  \hfill
  \raisebox{-0.5\height}{{\includegraphics[height=3.7cm]{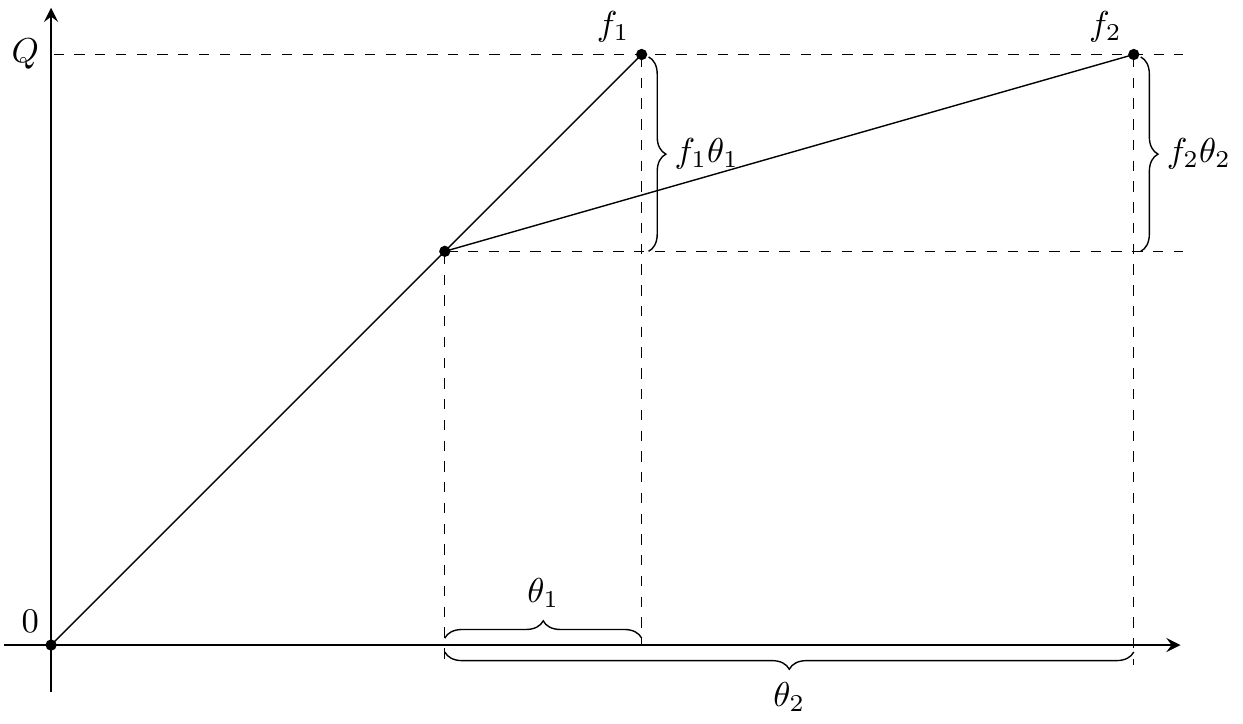}}}
  \hfill
  \raisebox{-0.5\height}{{\includegraphics[height=4.5cm]{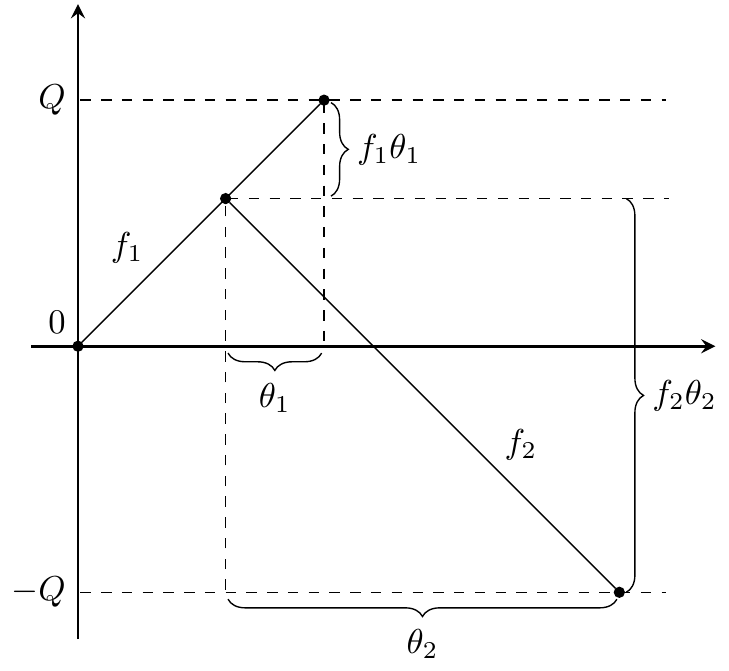}}	}
  \hfill\hspace*{0pt}
  \caption{Computing the updated firing date in the QSS
    method.\label{fig::qss2}}
\end{figure}
\enlargethispage{\baselineskip}

We illustrate the situation in the two diagrams of
Fig.~\ref{fig::qss2}, where we assume that $f_1$ is positive. For
instance, if the two slopes have the same sign (diagram to the left),
we need to update the firing date to the value $\theta_2$ such that
$|f_1| . \theta_1 = |f_2| . \theta_2$. Likewise, when $f_2$ is
negative, we have the relation $|f_1| . \theta_1 + |f_2| . \theta_2 =
2.Q$. Therefore, depending on the sign of $f_1.f_2$ (the sign of
$\dot{y}$ tell us whether $y$ is incremented or decremented) we
have $\DIF(t, x, y, \theta) = [\falpha(\theta),\falpha(\theta)]$ with:
\[
\falpha(\theta) = \frac{|f(x,y \pm Q)|}{|f(x,y)|} . \theta
\qquad \text{ or } \qquad \frac{|f(x,y \pm Q)|}{|f(x,y)|} . \left (
  2.Q.|f(x, y)| - \theta \right )
\]

This example shows
that it is possible to implement the QSS method using only linear
fickle functions. (We discuss briefly the associated class of \DTPN at
the end of Sect.~\ref{sec:state-class-abstr}.) Since the notion of
slope is central in our implementation of the QSS method, we could
have used instead an extension of \TPN with multirate
transitions~\cite{Daws95twoexamples}, that is a model where time
advance at different rate depending on the state. While the case
$f_1.f_2 > 0$ lends itself well to this extension, it is not so
obvious when the slopes have different signs. On the opposite, it
would be interesting to use fickle transitions as a way to mimic
multirate transitions.


\section{A State Class Abstraction for Dynamic \TPN}
\label{sec:state-class-abstr}

In this section, we generalize the state class abstraction method to
the case of \DTPN. A State Class Graph (\SCG) is a finite abstraction
of the timed transition system of a net that preserves the markings
and traces. The construction is based on the idea that temporal
information in states (the firing functions) can be conveniently
represented using systems of difference
constraints~\cite{Ramalingam95solvingdifference}. We show that the
\SCG faithfully abstract the semantics of a net when the dynamic
interval functions are translations. We only over-approximate the set
of reachable markings in the most general case.

A state class $C$ is defined by a pair $(m, D)$, where $m$ is a
marking and the firing domain $D$ is described by a (finite) system of
linear inequalities. We say that two state classes $C = (m, D)$ and
$C'= (m', D')$ are equal, denoted $C \cong C'$, if $m=m'$ and $D
\Leftrightarrow D'$ (i.e. $D$ and $D'$ have equal solution
sets). Hence class equivalence is decidable. In a domain $D$, we use
variables $\sI{t}, \ssI{t}, \dots$ to denote a constraint on the value
of $\tI[t]$. A domain $D$ is defined by a set of difference
constraints in reduced form: $\alpha_i \leq \sI{i} \leq \beta_i$ and $\sI{i} - \sI{j} \leq
 \gamma_{i,j}$,
where $i, j$ range over a given subset of ``enabled transitions'' and
the coefficients $\alpha, \beta$ and $\gamma$ are rational numbers. We
can improve the reduced form of $D$ by choosing the tightest possible
bounds that do not change its associated solutions set. In this case
we say that D is in closure form. We show in Th.~\ref{th:scg} how to
compute the coefficients of the closure form incrementally.


In the remainder of this section, we use the notation $\falphas_t(m)$
and $\fbetas_t(m)$ for the left and right endpoints of
$\SIF(t,m)$. Likewise, when the marking $m$ is obvious from the
context, we use the notations $\falpha_t(\theta)$ and
$\fbeta_t(\theta)$ for the left and right endpoints of
$\DIF(t,m,\theta)$, that is $\falpha_t(\theta) =
\low{\DIF(t,m,\theta)}$ and $\fbeta_t(\theta) =
\high{\DIF(t,m,\theta)}$. We call $\falpha_t$ and $\fbeta_t$ the
fickle functions of $t$. In the remainder of the text, we assume that
$0 \leq \falpha_t(\theta) \leq \fbeta_t(\theta)$ for all possible
(positive) date $\theta$ and that $\falpha(\infty) = \infty$. We also
require these functions to be monotonically increasing. We impose no
other restrictions on the fickle functions.


We define inductively a set of classes $C_\sigma$, where $\sigma \in
T^*$ is a sequence of discrete transitions firable from the initial
state. This is the State Class Graph construction
of~\cite{BM82,BD91}. Intuitively, the class $C_\sigma = (m, D_\sigma)$
``collects'' the states reachable from the initial state by firing
schedules of support sequence $\sigma$. The initial class
$C_{\epsilon}$ is $(m_{0}, D_0)$ where $D_0$ is the domain defined by
the set of inequalities $\falphas_i(m_0) \le \sI{i} \le
\fbetas_i(m_0)$ for all $i$ in ${\cal E}(m_0)$.

Assume $C_\sigma=(m,D)$ is defined and that $t$ is enabled at $m$. We
details how to compute the domain for the class $C_{\sigma.t}$. First
we test whether the system $D$ extended with the constraints $D_t =
\{\sI{k} - \sI{t} \ge 0 \mid t \neq k, k \in \E(m)\}$ is
consistent. This is in order to check that transition $t$ can be fired
before any other enabled transitions $k$ at $m$. If $D \wedge D_t$ is
consistent, we add $C_{\sigma.t} = (m',D')$ to the set of reachable
classes, where $m'$ is the result of firing $t$ from $m$, i.e.  $m' =
m - \Pre(t) + \Post(t)$. The computation of $D'$ follows the same
logic than with standard \TPN.

We choose a set of fresh variables, say $\ssI{k}$, for every
transition $k$ that is enabled at $m'$. For every persistent
transition, $k \in \pers(m,t)$, we add the constraints $\sI{k} =
\ssI{k}- \sI{t}$ to the set of inequalities in $D \wedge D_t$. The
variable $\ssI{k}$ matches the firing date of $k$ at the time $t$
fires, that is, the value of $\tI[k]$ used in the expression $\DIF(k,
m', \tI[k])$ (see Definition~\ref{def:sg}, case $(i)$). For every
newly enabled transition, $k \in \nenabl(m,t)$, we add the constraints
$\falphas_k(m') \le \ssI{k} \le \fbetas_k(m')$. This constraint
matches the fact that $\tI[k]'$ is in the interval $\SIF(k, m')$ if
$k$ is newly enabled at $m'$. As a result, we obtain a set of
inequations where we can eliminate all occurrences of the variables
$\sI{k}$ and $\sI{t}$. After removing redundant inequalities and
simplifying the constraints on transitions in conflicts with $t$---so
that the variables only ranges over transitions enabled at $m'$---we
obtain an ``intermediate'' domain $D_\mathit{int}$ that obeys the
constraints: $\kappa_{i} \leq \ssI{i} \leq \lambda_{i}$ and $\ssI{i} -
\ssI{j} \leq \mu_{i,j}$, where $i, j$ range over $\E(m')$ and the
constants $\kappa, \lambda$ and $\mu$ are defined as follows.
\[
\begin{array}[c]{lcl@{\quad}l}
  \kappa_i &=& 
  \left \{ \begin{array}[c]{l}
      \falphas_i(m')\\
      \max\, (0 , \{ -\gamma_{i,j} \mid {i,j \in \E(M)} \})\\ 
    \end{array} \right . & 
  \begin{array}[c]{l}
    \text{if $i$ is newly enabled,}\\
    \text{otherwise}
  \end{array}
  \\[1em]
  \lambda_i &=& 
  \left \{ \begin{array}[c]{l}
      \fbetas_i(m')\\
      \gamma_{i,t}\\
    \end{array} \right . &
  \begin{array}[c]{l}
    \text{if $i$ is newly enabled,}\\
    \text{otherwise}
  \end{array}
\\[1em]
  \mu_{i,j} &=& 
  \left \{ \begin{array}[c]{l}
      \lambda_i - \kappa_j\\
      \min\, (\gamma_{i,j}, \lambda_i - \kappa_j)\\
    \end{array} \right . &
  \begin{array}[c]{l}
    \text{if either $i$ or $j$ newly enabled,}\\
    \text{otherwise}\\
  \end{array}
\end{array} \tag{C1}
\]

Finally, we need to apply the effect of the fickle functions. For
this, we rely on the fact that $\falpha_i$ and $\fbeta_i$ are
monotonically increasing functions. To obtain $D'$, we choose a set of
fresh variables, say $\sI{i}'$, for every transition $i \in \E(m')$
and add the following relations to $D_\mathit{int}$. To simplify the
notation, we assume that in the case of a newly enabled transition,
$j$, the functions $\falpha_j$ and $\fbeta_j$ stand for the identity
function (with this shorthand, we avoid to distinguish cases where
both or only one of the transitions are persistent):
\[
\begin{array}[c]{l@{\quad}l}
  \sI{i}' = \ssI{i}   & \text{if $i,
    j$ are  newly enabled}\\
  \falpha_i(\ssI{i}) \leq \sI{i}' \leq \fbeta_i(\ssI{i}) \text{\ and\ } \sI{i}' -
  \sI{j}' \leq \fbeta_i(\ssI{i}) - \falpha_j(\ssI{j}) & \text{if $i$
    or $j$ are persistent}\\
\end{array}
\]
The relation for newly enabled transitions simply states that
$\ssI{i}$ already captures all the constraints on the firing time
$\tI[i]'$. For persistent transitions, the first relation states that
$\sI{i}'$ is in the interval $[\falpha_i(\ssI{i}),
\fbeta_i(\ssI{i})]$, that is in $\DIF(i, m', \tI(i))$.

We obtain the domain $D'$ by eliminating all the variables of the kind
$\ssI{i}$.  First, we can observe that, by monotonicity of the
functions $\falpha_i$ and $\fbeta_i$, we have $\falpha_i(\kappa_i)
\leq \falpha_i(\ssI{i})$ and $\fbeta_i(\ssI{i}) \leq
\fbeta_i(\lambda_i)$. This gives directly a value for the coefficients
$\alpha'_i$ and $\beta'_i$. The computation of the coefficient
$\gamma'_{i,j}$ is more complex, since it amounts to computing the
maximum of a function over a convex sets of points. Indeed
$\gamma'_{i,j}$ is the least upper-bound for the values of $\sI{i}' -
\sI{j}'$ over $D_{\mathit{int}}$ or, equivalently:
\[
\begin{array}{lcl}
  \gamma'_{i,j} &=& \max \, \{ 
  \fbeta_i(\ssI{i}) - \falpha_j(\ssI{j}) \mid \ssI{i},
  \ssI{j} \in D_{\mathit{int}} \}\\
  &=& \max \, \{ \fbeta_i(\ssI{i}) - \falpha_j(\ssI{j}) \mid
  \kappa_i \leq \ssI{i} \leq \lambda_i,
  \kappa_j \leq \ssI{j} \leq \lambda_j,
  \ssI{i} - \ssI{j} \leq \mu_{i,j}\}\\
  \end{array}
\]

It is possible to simplify the definition of $\gamma'_{i,j}$. Indeed,
if we fix the value of $\ssI{j}$ then, by monotonicity of $\fbeta_i$,
the maximal value of $\fbeta_i(\ssI{i}) - \falpha_j(\ssI{j})$ is
reached when $\ssI{i}$ is maximal. Hence we have two possible cases:
either $(i)$ it is reached for $\ssI{i} = \ssI{j} + \mu_{i,j}$ if
$\kappa_j \leq \ssI{j} \leq \lambda_i + \mu_{i,j}$; or $(ii)$ it is
reached for $\ssI{i} = \lambda_i$ if $\lambda_i - \mu_{i,j} \leq
\ssI{j} \leq \lambda_j$. This result is illustrated in the schema of
Fig.~\ref{fig::dbm::gamma::int}, where we display an example of domain
$D_\mathit{int}$. When $\ssI{j}$ is constant (horizontal line), the
maximal value is on the ``right'' border of the convex set (bold
line). We also observe that in case $(ii)$, by monotonicity of
$\falpha_j$, the maximal value is equal to $\fbeta_i(\lambda_i) -
\falpha_j(\lambda_i + \mu_{i,j})$. Therefore the value of
$\gamma'_{i,j}$is obtained by computing the maximal value of the
expression $\fbeta_i(\theta) - \falpha_j(\theta - \mu_{i,j})$, that
is:
\[
  \gamma'_{i,j} = \max \, \{ \fbeta_i(\theta) - \falpha_j(\theta - \mu_{i,j}) \mid   
  \kappa_j + \mu_{i,j} \leq \theta \leq \lambda_i \}  \tag{C2}
\]
As a consequence, the value of $\gamma'_{i,j}$ can be computed by
finding the minimum of a numerical function (of one parameter) over a
real interval, which is easy.

\begin{figure}
  \centering
  \subfloat[Domain $D_{\mathit{int}}$ projected over $t_i, t_j$
  \label{fig::dbm::gamma::int}] {
    \includegraphics[scale=0.95]{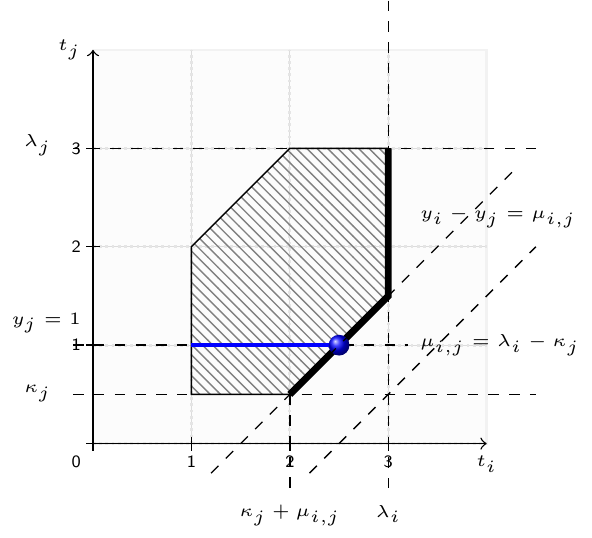}
  }
  \subfloat[Domain $D'$ obtained from $D_{\mathit{int}}$
  \label{fig::dbm::gamma::prime}] {
    \includegraphics[scale=0.95]{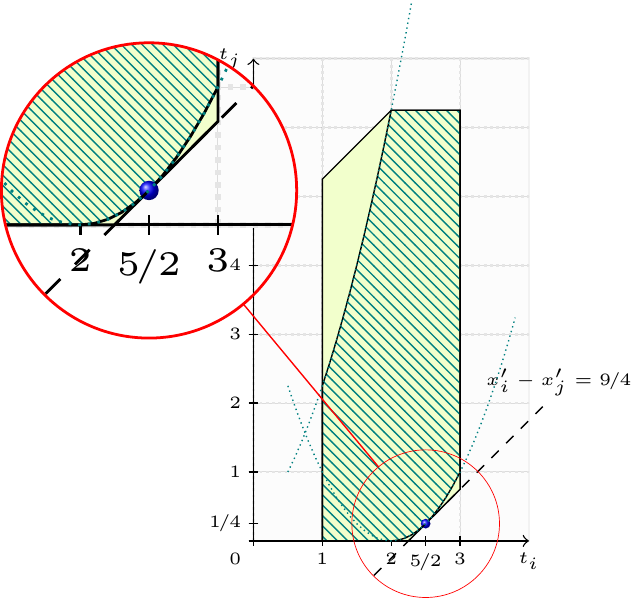}
  }
  \caption{Computing the coefficient $\gamma'_{i,j}$ in the domain $D'$.} \label{fig::dbm::gamma}
\end{figure}

We display in Fig.~\ref{fig::dbm::gamma::prime} the domain $D'$
obtained from $D_\mathit{int}$ after applying the fickle functions. In
this example, $t_j$ is the only fickle transition and we choose
$\falpha_j(\theta) = \fbeta_j(\theta) = (\theta - \half)^2$ when
$\theta \geq \half$. With our method we have that $\mu_{i,j} =
\nicefrac{3}{2}$ and the value of $\gamma'_{i,j}$ is obtained by
computing the maximal value of the expression $(\theta - \half)^2 -
(\theta - \nicefrac{3}{2})$ with $\theta \in [2,3]$, that is
$\nicefrac{9}{4}$.

\begin{theorem}\label{th:scg}
  Assume $C = (m, D)$ is a class with $D$ in closure form. Then for
  every transition $t$ in $\E(m)$ there is a unique class $(m', D')$
  obtained from $C$ by firing $t$. The domain $D'$ is also in closure
  form and can be computed incrementally as follows (we assume that
  $\falpha_i$ and $\fbeta_i$ stands for the identity functions when
  $i$ is newly enabled).
  \[
  \begin{array}[c]{lcl@{\quad}l}
    \alpha'_i 
    &=& \falphas_i(m') & \text{if $i$ is newly enabled,}\\
    &=& \max\, (\falpha_i(0) , 
    \{ \falpha_i(-\gamma_{i,j}) \mid {i,j \in \E(M)} \}) &
    \text{otherwise}\\ 

    \beta'_i 
    &=& \fbetas_i(m') & \text{if $i$ is newly enabled,}\\
    &=& \fbeta_i(\gamma_{i,t}) & \text{otherwise}\\

    \gamma'_{i,j} 
    &=& \min (\gamma_{i,j}, \beta'_i - \alpha'_j) &
    \text{if $i,j$ are newly enabled,}\\
    &=&  \max \{ \fbeta_i(\theta) -
    \falpha_j(\theta - \mu_{i,j}) \mid \mu_{i,j} + \kappa_j \leq x \leq
    \lambda_i\}
    & \text{otherwise}\\
    \multicolumn{4}{r}{(\text{where } \lambda_i,
      \kappa_j \text{ and } \mu_{i,j} \text{ are defined as in \textrm{(C1)}}\/)}\\
  \end{array}
  \]
  Moreover, if the state $(m, \tI)$ is reachable in the state graph of
  a net, say $N$, and $(m, \tI) \trans{\theta}\trans{t} (m', \tI')$
  then there is a class $C_\sigma = (m, D)$ reachable in the \SCG
  computed for $N$ with $\tI \in D$, $C_{\sigma.t} = (m', D')$ and
  $\tI'\in D'$.
\end{theorem}

The hatched area inside the domain displayed in
Fig.~\ref{fig::dbm::gamma::prime} is the image of the domain
$D_\mathit{int}$ after its transformation by the fickle function
$\falpha_j(\theta)$. We see that some points of $D'$ have no
corresponding states in $D_\mathit{int}$. Hence we only have an
over-approximation. (We do not have enough place to give an example of
net with a marking that is in reachable in the \SCG but not reachable
in the state space, but such an example is quite easy to build.)  If
we consider the definition of the coefficients $\gamma'$ in equation
(C2), we observe that the situation is much simpler if the fickle
functions are translations. Actually, it is possible to prove that, in
this case, the \SCG construction is exact.

\begin{theorem}\label{th:regular}
  If the \DTPN $N$ is a translation then the \SCG defined in
  Th.~\ref{th:scg} has the same set of reachable markings and the same
  set of traces than the timed transition system of $N$.
\end{theorem}
\begin{proof}[sketch]
  By equation (C2), if the net is a translation then there are two
  constants $c_i, c_j$ such that $\fbeta_i(\theta) = \theta + c_i$ and
  $\falpha_j(\theta) = \theta + c_j$. Therefore the expression
  $\fbeta_i(\theta) - \falpha_j(\theta - \mu_{i,j})$ is constant and
  equal to $c_i - c_j - \mu_{i,j}$ (the maximum is reached all over
  the boundary of the domain). In this case, every state in $D'$ has a
  corresponding states in $D_\mathit{int}$.
\end{proof}

We can also observe that, if the dynamic interval bounds $\falpha_i$
and $\fbeta_i$ are linear functions, then we can follow a similar
construct using (general) systems of inequations for the domains
instead of difference constraints. This solution gives also an exact
abstraction for the state space but is not interesting from a
computational point of view (since we loose the ability to compute a
canonical form for the domain incrementally). In this case, we are in
a situation comparable to the addition of stopwatch to \TPN where
systems of difference constraints are not enough to precisely capture
state classes. With our computation of the coefficient $\gamma'$, see
equation (C2), we use instead the ``best difference bound matrix''
that contains the states reachable from the class $C$. This
approximation is used in some tools that support stopwatches, like
Romeo~\cite{gardey2005romeo}.

\section{Two Application for Dynamic \TPN}
\label{sec:two-appl-dynam}

We study two possible applications for fickle transitions. First to
model a system of preemptive, periodic tasks with fixed duration. Next
to model hybrid system with non trivial continuous behavior. These
experiments have been carried out using a prototype extension of Tina.
The tool and the all the models are available online at
\url{http://projects.laas.fr/tina/fickle/}.

\subsection{Scheduling Preemptive Tasks}
\label{sec:sched-preempt-tasks}

We consider a simple system consisting of two periodic tasks, Task1
and Task2, executing on a single processor. Task2 has a period of $10$
unit of time (u.t.) and a duration of $6$ u.t. ; Task1 has a period of
$5$ u.t. and a duration of $1$ and can preempt Task2 at any time. We
display in Fig.~\ref{fig::example::sched} a \TPN model for this
system. Our model makes use of a \emph{stopwatch arc}, drawn using an
``open box'' arrow tip
($\rule[2pt]{1em}{0.5pt}\mkern-3mu{\boxempty}$), and of an inhibitor
arc ($\rule[2pt]{1em}{0.5pt}\mkern-3mu{\ocircle}$).

The net is the composition of four components. The roles of Sched1 and
Sched2 is is to provide a token in place \ndname{psched} at the
scheduling date of the tasks. The behavior of the nets corresponding
to Task1 and Task2 are similar. 
Both nets are $1$-safe and their (unique) token capture the state of
the tasks. When the token is in place \ndname{e}, the task execute;
when it is in place \ndname{w} it is waiting for its scheduling
event. Hence we have a scheduling error if there is a token in place
\ndname{psched} and not in place \ndname{w}.

\begin{figure}
\centering
\begin{tikzpicture}[      
  every node/.style={anchor=south west,inner sep=0pt},
  x=1mm, y=1mm,
  ]   

  \tikzstyle{marks}=[thick,dashed,rounded corners,draw=teal]
  \tikzstyle{nom}=[at start,below right,yshift = -2pt,xshift=2em]
  \node (fig1) at (0,0)
  {\includegraphics[width=9.7cm]{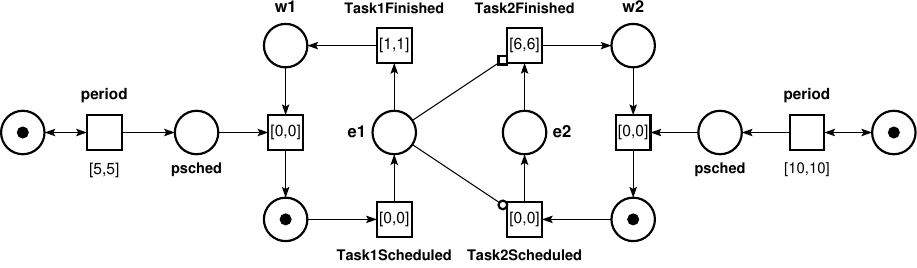}};

  \draw[marks] (-1,-0.5) rectangle (25,30)
  node[nom]{Sched1};
  \draw[marks] (26,-0.5) rectangle (49,30)
  node[nom]{Task1};
  \draw[marks] (50,-0.5) rectangle (72,30)
  node[nom]{Task2};
  \draw[marks] (73,-0.5) rectangle (100,30)
  node[nom]{Sched2};
\end{tikzpicture}
\caption{System with one preemptive and one simple
  task.} \label{fig::example::sched}
\end{figure}
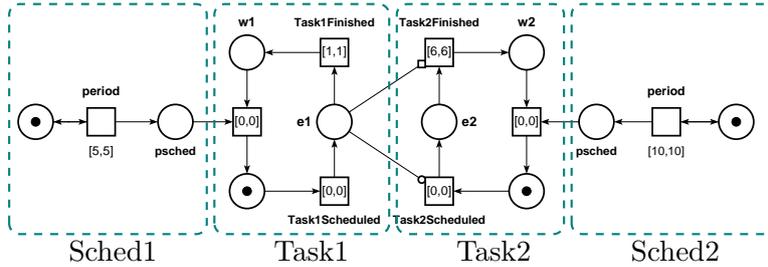

We use an inhibitor arc between the place \ndname{e1} and the
transition \ndname{Task2Scheduled} to model the fact that Task2 cannot
use the processor if Task1 is already running. We use a stopwatch arc
between \ndname{e1} and the transition \ndname{Task2Finished} to model
the fact that Task1 can preempt Task2 at any moment. A stopwatch
(inhibitor) arc ``freezes'' the firing date of its
transition. Therefore the completion of Task2 (the firing date of
\ndname{Task2Finished}) is postponed as long as Task1 is
running. Using the same approach, we can define a \TPN modeling a
system with one preemptive task and $n$ ``simple''tasks.

We can define an equivalent model using fickle transitions instead of
stopwatch. The idea is to add the duration of Task1 to the completion
date of Task2 each time Task1 starts executing (that is
\ndname{Task1Scheduled} fires). This can be obtained by removing
stopwatch arcs and using for \ndname{Task2Finished} the fickle
functions $\falpha(\theta) = \fbeta(\theta) = \theta + 1$ when
\ndname{Task1Scheduled} fires and the identity otherwise.  The
resulting dynamic \TPN is a translation and therefore the \SCG
construction is exact. In this new model, we simulate preemption by
adding the duration of the interrupting thread to the completion date
of the other running thread. The same idea was used by Bodeveix et
al. in~\cite{NaRaBoFi2008}, where they prove the correctness of this
approach using the B method. This scheme can be easily extended to an
arbitrary set of preemptive tasks with fixed priority.

The following table gives the results obtained when computing the \SCG
for different number of tasks. The models with fickle transitions have
slightly more classes than their stopwatch counterpart. Indeed, in the
fickle case, the firing date of \ndname{Task2Finished} can reach a
value of $7$, while it is always bounded by $6$ with stopwatches. The
last row of the Table gives the computation time speedup between our
implementation of fickle transitions and the default implementation of
stopwatch in Tina. We observe that the computation with fickle
transitions is (consistently) two times faster; this is explained by
the fact that the algorithmic for stopwatches is more complex. Memory
consumption is almost equal between the two versions approaches, with
a slight advantage for the fickle model.

\newlength{\sdzlength}
\setlength{\sdzlength}{0.8\textwidth}
\newcommand{\rrq}[2]{\DIVIDE{#2}{#1}{\ratio}\ROUND{\ratio}{\ratio}$\times
  \ratio$ $\displaystyle\left ( \nicefrac{#1 s}{#2 s} \right )$}
\newcommand{\rrs}[2]{%
\begin{tabular}[b]{@{}p{0.15\sdzlength}@{}}
  \hfill${{#1}}$\hfill\hspace*{0pt}\\
  \hline
  \hfill${{#2}}$\hfill\hspace*{0pt}
\end{tabular}}
\newcommand{\firstcol}[2]{%
  \begin{tabular}[c]{@{}c@{}}
    {#1}\\
    {#2}
  \end{tabular}}
\begin{table}
\begin{center}
  \begin{tabular}{|>{\centering\arraybackslash}p{0.2\sdzlength}| 
      *5{>{\centering\arraybackslash}m{0.155\sdzlength}|} @{}m{0pt}@{}}    
    \hline
    {\# tasks} & 2 &  3 &  5
    &  10 &  12 
    &\\
    \hline
    \firstcol{\# states} 
    {$\displaystyle\left ( \nicefrac{\text{fickle}}{\text{stopwatch}} \right )$}  & 
    \rrs{84}{83} & 
    \rrs{208}{205} & 
    \rrs{1\,786}{1\,771} & 
    \rrs{539\,902}{539\,391} & 
    \rrs{5\,447\,504}{5\,445\,457} 
    &\\
    \hline
    \firstcol{time speedup}
    {$\displaystyle\left ( \nicefrac{\text{fickle}}{\text{stopwatch}} \right )$}  & 
    \rrq{0.005}{0.010} & 
    \rrq{0.022}{0.042}& 
    \rrq{0.37}{0.784} & 
    \rrq{170}{392} & 
    \rrq{3077}{6024} 
    &\\
    \hline
  \end{tabular}\\
\end{center}
\caption{Comparing the use of Fickle Transitions and Stopwatch.}
  \label{tab:sched-results}
\end{table}

\subsection{Verification of Linear Hybrid systems}
\label{sec:other-examples}

The semantics of fickle transitions came naturally from our goal of
implementing the QSS method using \TPN (see
Sect.~\ref{sec:simple-example-dynam}). We give some experimental
results obtained using this approach on two very simple use cases.

Our first example is a model for the behavior of hydraulic cylinders
in a {landing gear system}~\cite{ABZ2014-Landing-Gear-System-Case-Study}.
The system can switch between two modes, extension and retraction. The
only parameter is the position $x$ of the cylinder head. (It is
possible to stop and to inverse the motion of a cylinder at any time.)
The system is governed by the relation $\dot{x} = 5 - x$ while
opening, with $x \in [0, 5]$, and $\dot{x} = -1$ while closing. We can
model this system using two fickle transitions.

\begin{figure}
  \centering
  \hfill
  \raisebox{-0.5\height}{{\includegraphics[width=0.4\textwidth]{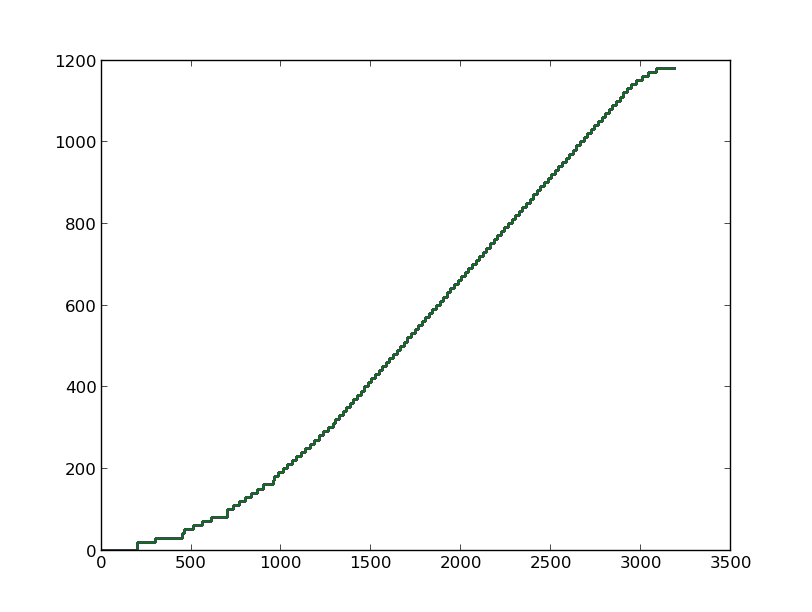}}}
  \hfill
  \raisebox{-0.5\height}{{\includegraphics[width=0.4\textwidth]
    {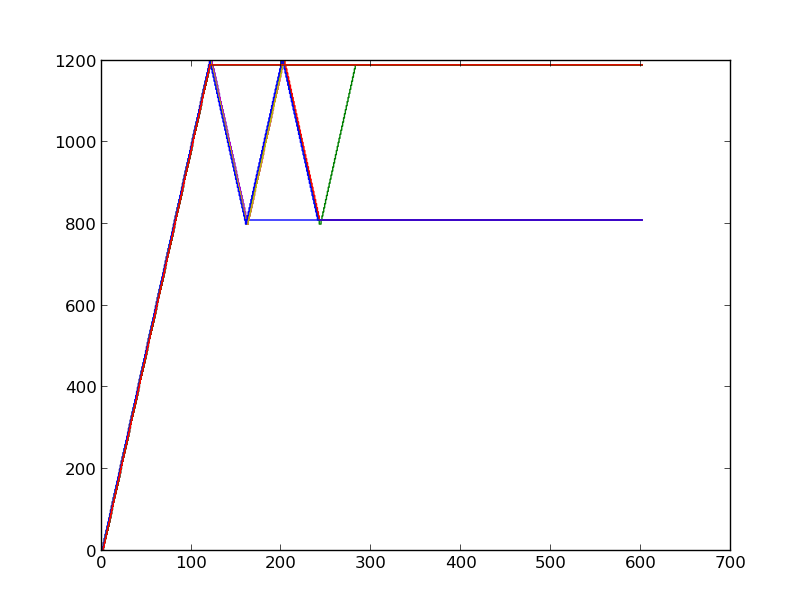}}}
  \hfill\hspace*{0pt}
  \caption{Evolution of the PI-controller: fickle (left) and
    discrete (right) versions.\label{fig::pid}}
\end{figure}

The second example is a model for a \emph{double integrator}, an
extension of the simple integrator of Fig.~\ref{fig::qss} to a system
with two interdependent variables $x_1$ and $x_2$. The system has two
components, $P_1, P_2$, where $P_i$ is in charge of the evolution of
$x_i$, for $i \in \{1,2\}$, and each $x_i$ is governed by the relation
$\dot{x_i} = f_i(x_1, x_2)$. The components $P_1$ and $P_2$ are
concurrent and interact with each other by sending an event when the
value of $x_i$ changes. Therefore the system mixes message passing and
hybrid evolution. This system can be used to solve second order linear
differential equations of the form $\ddot{y} = k_P \dot{y} + k_I (S -
y)$; we simply take $\dot{x_1} = x_2$ and $\dot{x_2} = k_P x_2 + k_I
(S - x_1)$. This family of equations often appears in control-loop
feedback mechanisms, where they model the behavior of
proportional-integral (PI) controller. For example, a system with
double quantized integrators is studied in~\cite{foures2012formal} in
the context of a dynamic cruise controller.

We compare the results obtained with our two versions of the
integrator: fickle and discrete (synchronous).  Figure.~\ref{fig::pid}
displays the evolution of the variable $x_1$ in the PI-controller for
our two models, with a quantum of $\nicefrac{1}{10}$. We observe that
the discrete version does not converge with this time step (we need to
choose a value of $\nicefrac{1}{100}$).



\begin{center}
  \begin{tabular}{|>{\centering\arraybackslash}m{0.18\sdzlength}|
      >{\centering\arraybackslash}m{0.14\sdzlength}|
      >{\centering\arraybackslash}m{0.14\sdzlength}|
      >{\centering\arraybackslash}m{0.16\sdzlength}|
      >{\centering\arraybackslash}m{0.16\sdzlength}|
      >{\centering\arraybackslash}m{0.16\sdzlength}|
      @{}m{0pt}@{}}    
    \hline
    {System} & 
    \multicolumn{2}{c|}{Landing Gear} &
    \multicolumn{3}{c|}{ Cruise Control (PI-controller)} &\\
    \cline{2-6}
    {(version) parameters} &
    (fickle) $Q=\nicefrac{1}{10}$  & 
    (discrete) $h=\nicefrac{1}{10}$ &  
    (fickle) $Q = \nicefrac{1}{10}$  &  
    (discrete) $h=\nicefrac{1}{10}$ &
    (discrete) $h=\nicefrac{1}{100}$  &\\
    \hline
    {\# states} 
    & 
    {1\,906} & 
    {2\,590} & 
    {259} &
    {2\,049} &
    {20\,549} &\\
    \hline
    {time (s)}
    & 
    {0.076} & 
    {0.125} & 
    {0.004} &
    {0.017} &
    {0.185} &\\
    \hline
    {memory\,(MB)}
    & 
    {1.00} & 
    {1.56} & 
    {0.11} &
    {0.90} &
    {9.02} &\\
    \hline
  \end{tabular}\\
\end{center}

\section{Conclusion and Related Work}
\label{sec:conclusion}

We have shown how to extend the \SCG construction to handle fickle
transitions. The \SCG is certainly the most widely used state space
abstraction for Time Petri nets: it is a convenient abstraction for
LTL model checking; it is finite when the set of markings is bounded;
and it preserves both the markings and traces of the net. The results
are slightly different with dynamic \TPN, even for the restricted
class of translation nets. In particular, we may have an infinite \SCG
even when the net is bounded. This may be the case, for instance, if
we have a transition that can stay persistent infinitely and that is
associated to the fickle function $\DIF(\theta) = [\theta + 1,
\theta+1]$. This entails that our construction may not terminate, even
if the set of markings is bounded. This situation is quite comparable
to what occurs with \emph{updatable timed
  automata}~\cite{Bouyer2004291} and, like in this model, it is
possible to prove that the model-checking problem is undecidable in
the general case. This does not mean that our construction is useless
in practice, as we show in our examples of
Sect.~\ref{sec:two-appl-dynam}.

The notion of fickle transitions came naturally as the simplest
extension of \TPN able to integrate the Quantized State System (QSS)
method~\cite{cellier06:_contin_system_simul} inside Tina. Although
there are still problems left unanswered, this could provide a
solution for supporting hybrid systems inside a real-time
model-checker. Theorem~\ref{th:scg} gives clues on how to support
fickle transitions in existing tools for standard \TPN. Indeed, the
incremental computation of the coefficients of the ``difference-bound
matrices'' ($\alpha, \beta$ and $\gamma$) is not very different from
what is already implemented in tools that can generate a \SCG for a
standard \TPN. In particular, the ``intermediate'' domain
$D_\mathit{int}$ computed in (C1) is exactly the domain obtained from
$D$ in a standard \TPN. We only need two added elements. First, we
need to apply a numerical function over the coefficients of
$D_\mathit{int}$; this is easy if the tool already supports
associating a function to a transition in a \TPN (as it is the case
with Tina). Next, we need to compute the maximal value of a numerical
functions over a given interval; this can be easily added to the tool
or delegated to a numerical solver. Actually, for the examples
presented in Sect.~\ref{sec:two-appl-dynam}, we only need to use
affine functions, for which the maximal value can be defined by a
straightforward arithmetical expression.  As a result, it should be
relatively easy to adapt existing tools to support the addition of
fickle transitions. This assessment is supported by our experience
when extending Tina; once the semantics of fickle transitions was
stable, it took less than a week to adapt our tools and to obtain the
first results.

To our knowledge, updatable \TA is the closest model to dynamic \TPN.
The relation between these two models is not straightforward. We
consider very general update functions but do not allow the use of
multiple firing dates in an update (that would be the equivalent of
using other clocks in \TA). Also, the notion of persistent transitions
does not exist in \TA while it is central in our approach.  While the
work on updatable \TA is geared toward decidability issues, we rather
concentrate on the implementation of our extension and its possible
applications. Nonetheless, it would be interesting to define a formal,
structural translations between the two models, like it was done
in~\cite{Cassez05comparisonof,BPV06} between \TA and \TPN.
Some of our results also show similarities between fickle transitions
and the use of stopwatch~\cite{BLRV06}. In the general case, it does
not seem possible to encode one extension with the other, but it would
be interesting to look further into this question. Finally, since the
notion of slope is central in our implementation of the QSS method
(see Sect.~\ref{sec:simple-example-dynam}), it would be interesting to
compare our results with an approach based on multirate
transitions~\cite{Daws95twoexamples}, that is a model where time does
not advance at the same rate in all the transitions.

For future works, we plan to study an extension of our approach to
other models of real-time systems and to other state-space
abstractions. For instance the Strong \SCG construction
of~\cite{Ber03b}, that is finer than the \SCG construction but that is
needed when considering the addition of priorities. The strong \SCG
relies on the use of \emph{clock domains}, rather than firing domains,
and has some strong resemblance with the zone constructions commonly
used for analysis of \TA. Another, quite different, type of
abstractions rely on the use of a discrete time semantics for \TPN. We
can obtain a discrete semantic by, for instance, restricting
continuous transitions $\trans{\theta}$ to the case where $\theta$ is
an integer. This approach could be useful when modeling hybrid
systems, since it is a simple way to add a quantization over time as
well as over values.

\newpage

\bibliographystyle{plain}
\bibliography{main}

\begin{thebibliography}{10}

\bibitem{Ber03b}
{B. Berthomieu} and {F. Vernadat}.
\newblock State class constructions for branching analysis of {T}ime {P}etri
  {N}ets.
\newblock In {\em T{ACAS}2003}, volume LNCS2619, page 442. Springer, 2003.

\bibitem{Cassez05comparisonof}
B.~B{\'e}rard and F.~Cassez.
\newblock Comparison of the expressiveness of timed automata and time petri
  nets.
\newblock In {\em Proc. FORMATS’05, vol. 3829 of LNCS}, pages 211--225.
  Springer, 2005.

\bibitem{BD91}
B.~Berthomieu and M.~Diaz.
\newblock Modeling and verification of time dependent systems using time
  {P}etri nets.
\newblock {\em IEEE Trans. on Software Engineering}, 17(3):259--273, 1991.

\bibitem{BLRV06}
B.~Berthomieu, D.~Lime, O.H. Roux, and F.~Vernadat.
\newblock Reachability problems and abstract state spaces for time {Petri} nets
  with stopwatches.
\newblock {\em Journal of Discrete Event Dynamic Systems}, 17:133-158, 2007.

\bibitem{BM82}
B.~Berthomieu and M.~Menasche.
\newblock A state enumeration approach for analyzing time {P}etri nets.
\newblock In {\em Proc. Applications and Theory of {P}etri Nets (ATPN'82)},
  pages 27--56, 1982.

\bibitem{BPV06}
B.~Berthomieu, F.~Peres, and F.~Vernadat.
\newblock Bridging the gap between timed automata and bounded time petri nets.
\newblock In {\em Formal Modeling and Analysis of Timed Systems (FORMATS'06),
  Springer LNCS 4202}, pages 82--97, 2006.

\bibitem{BRV04}
B.~Berthomieu, P.-O. Ribet, and F.~Vernadat.
\newblock The tool {TINA} -- construction of abstract state spaces for {P}etri
  nets and time {P}etri nets.
\newblock {\em International Journal of Production Research},
  42(14):2741--2756, 15 July 2004.

\bibitem{ABZ2014-Landing-Gear-System-Case-Study}
Fr\'ed\'eric Boniol and Virginie Wiels.
\newblock {T}he {L}anding {G}ear {S}ystem {C}ase {S}tudy.
\newblock In {\em ABZ Case Study}, volume 433 of {\em Communications in
  Computer Information Science}. Springer, 2014.

\bibitem{Bouyer2004291}
Patricia Bouyer, Catherine Dufourd, Emmanuel Fleury, and Antoine Petit.
\newblock Updatable timed automata.
\newblock {\em Theoretical Computer Science}, 321(2–3):291--345, 2004.

\bibitem{cellier06:_contin_system_simul}
Francois~E. Cellier and Ernesto Kofman.
\newblock {\em Continuous System Simulation}.
\newblock Springer, 2006.

\bibitem{cellier2008quantized}
Fran{\c{c}}ois~E Cellier, Ernesto Kofman, Gustavo Migoni, and Mario Bortolotto.
\newblock Quantized state system simulation.
\newblock {\em Proc. GCMS’08, Grand Challenges in Modeling and Simulation},
  pages 504--510, 2008.

\bibitem{Daws95twoexamples}
C.~Daws and S.~Yovine.
\newblock Two examples of verification of multirate timed automata with kronos.
\newblock In {\em Proc. 1995 IEEE Real-Time Systems Symposium, RTSS'95}, pages
  66--75. IEEE Computer Society Press, 1995.

\bibitem{foures2012formal}
Damien Foures, Vincent Albert, and Alexandre Nketsa.
\newblock Formal compatibility of experimental frame concept and {FD-DEVS}
  model.
\newblock {\em Proc. of MOSIM'12, International Conference of Modeling,
  Optimization and Simulation}, 2012.

\bibitem{gardey2005romeo}
Guillaume Gardey, Didier Lime, Morgan Magnin, and Olivier~H Roux.
\newblock Romeo: a tool for analyzing time petri nets.
\newblock In {\em Proceedings of the 17th international conference on Computer
  Aided Verification}, pages 418--423. Springer, 2005.

\bibitem{Merlin74astudy}
P.~M. Merlin.
\newblock {\em A study of the recoverability of computing systems}.
\newblock PhD thesis, Department of Information and Computer Science,
  University of California, 1974.

\bibitem{NaRaBoFi2008}
Odile Nasr, Miloud Rached, Jean-Paul Bodeveix, and Mamoun Filali.
\newblock {Sp{\'e}cification et v{\'e}rification d'un ordonnanceur en B via les
  automates temporis{\'e}s}.
\newblock {\em L'Objet}, 14(4), 2008.

\bibitem{Ramalingam95solvingdifference}
G.~Ramalingam, J.~Song, L.~Joscovicz, and R.~E. Miller.
\newblock Solving difference constraints incrementally.
\newblock {\em Algorithmica}, 23, 1995.

\bibitem{DBLP:journals/tse/VicarioSC09}
Enrico Vicario, Luigi Sassoli, and Laura Carnevali.
\newblock Using stochastic state classes in quantitative evaluation of
  dense-time reactive systems.
\newblock {\em IEEE Trans. Software Eng.}, 35(5):703--719, 2009.

\end{thebibliography}

\newpage

\appendix


\section{Proof of Theorem 1}
\label{sec:proof-theorem-1}

\noindent\textbf{Theorem~\ref{th:standard}:}
\textit{For every weak \DTPN, $N$, with a finite set of reachable
  markings, there is a \TPN, $N^\times$, with an equivalent
  semantics.}\\

We say that two nets have equivalent semantics if their state graphs
are weakly timed bisimilar (see Def.~\ref{def:bisim} below).

We assume that $N$ is the weak \DTPN $\langle
{P},{T},{\Pre},{\Post},m_0,\SIF, \DIF \rangle$. Since $N$ is weak, the
function \DIF is trivial and the behavior of persistent transitions is
the same than for \TPN. By hypothesis, we also have that the set of
markings of $N$, say $\Marks$, is bounded.

We define a 1-safe \TPN $N^\times$ that will simulate the execution of
$N$. Some places of $N^\times$ will be used to denote the marking in
the net $N$. We denote $P_{\Marks}$ the set containing one place for
every marking in \Marks. We use the same symbol, $m$, to denote the
place and the marking. The places of $P_{\Marks}$ are a subset of the
places of $N^\times$.

Since a \TPN is also an example of weak \DTPN, our construction can be
used in order to find a 1-safe \TPN equivalent to any given (bounded)
\TPN.

\begin{corollary}
  For every \TPN, with a finite set of reachable markings, there is a
  1-safe \TPN with an equivalent semantics.
\end{corollary}

The definition of $N^\times$ is based on the composition of a
collection of \TPN, denoted $E(t, I)$, that models the situation where
the transition $t$ of $N$ is currently enabled and where the firing
date of $t$ was picked in the time interval $I$. Therefore $I$ belongs
to the set of time intervals, denoted $\IM$, that can appear during
the evolution of $N$. The set $\IM$ is also finite and has less than
$|T|.|\Marks|$ elements.
\[
\IM = \left \{\ \SIF(k, m) \divides m \in \Marks, k \in T \ \right \}
\]
When dealing with a particular transition $t$, we can restrict to time
intervals of the form $\SIF(t,m)$ where $t$ is enabled at $m$.
\[
\Itrv(t) = \left \{\ \SIF(t, m) \divides t \in {\cal E}(m) \ \right \}
\]
Before defining formally $N^\times$, we start by defining some useful
notations and by giving some intuitions on our encoding.

\subsection{Definitions and Useful Notations}
\label{sec:useful-notations}

We define the \TPN $E(t, I)$ for every pair $(t, I)$ of a transition
$t$ in $T$ and a time interval $I$ in $\Itrv(t)$. We give a graphical
description of the net in Fig.~\ref{fig::appa}.

The net $E(t, I)$ has two places $\pp{t}{I}$ and $\pq{t}{I}$. The
place $\pp{t}{I}$ is the \emph{initial place} of
$E(t,I)$. Intuitively, we will place a token in $\pp{t}{I}$ when the
transition $t$ becomes newly-enabled by a marking, say $m$, and $I =
\SIF(t,m)$. The token moves to $\pq{t}{I}$ when the transition has
been enabled for long enough, that is when we reach the firing date of
$t$. Hence the purpose of transition $t_I$ is to record the timing
constraint associated to $t$. This transition is ``local'' to $E(t,
I)$, meaning that no other places in $N^\times$ has access to it.

\begin{figure}
  \centering
  {\includegraphics[width=\textwidth]{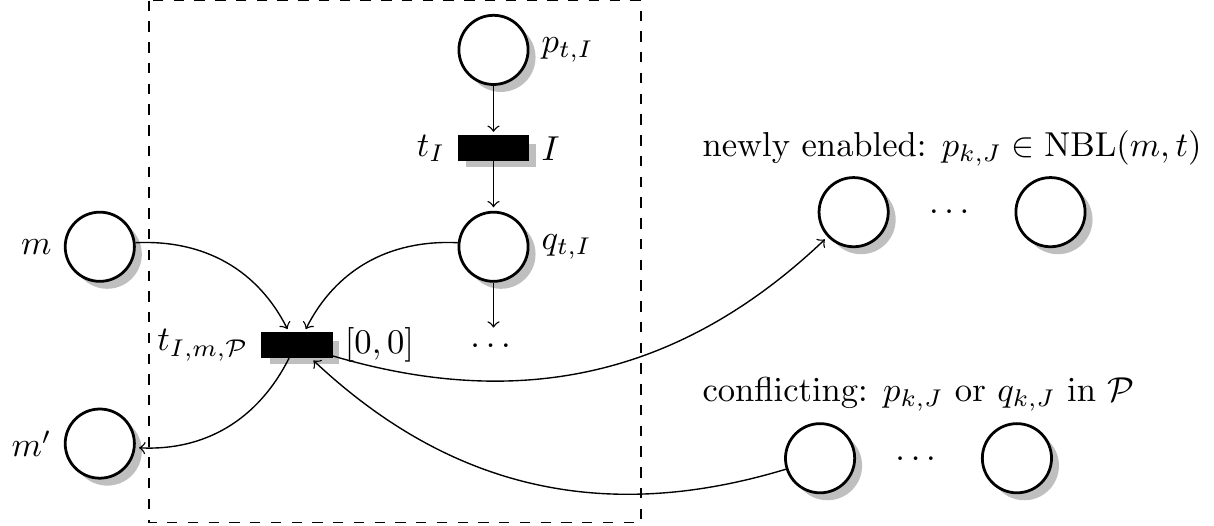}}
  \caption{The \TPN $E(t,I)$\label{fig::appa}}
\end{figure}

The final ingredient in the definition of $E(t, I)$ is a collection of
transitions $t_{I,m,{\cal P}}$, where the marking $m$ enables $t$, that
is $t$ is in ${\cal E}(m)$. These transitions have timing constraints
$[0,0]$ and can empty the token in place $\pq{t}{I}$. The purpose of
the transition $t_{I,m,{\cal P}}$ is to model the firing of transition
$t$ from the marking $m$ in $N$. In particular, a transition
$t_{I,m,{\cal P}}$ will empty the place $m$ of $P_{\Marks}$ and put a
token in the place $m'$ such that $m \trans{t} m'$, that is $m' = m -
\Pre(t) + \Post(t)$.

When the transition $t_{I,m,{\cal P}}$ fires, it should also
``enable'' new transitions and ``disable'' the transitions that are in
conflict with $t$ in $N$. More precisely, the transition should: (1)
put a token on the initial place of the net $E(k, J)$, where $k \in
\nenabl(m,t)$ and $J = \SIF(k, m')$; and (2) remove the token from the
net $E(k', J')$ such that $k'$ was enabled at $m$ but not at $m'$
(conflicting transitions). The transitions of $N$ that are persistent
when $t$ fires are not involved; therefore their firing date are left
untouched. The treatment of conflicting transitions is quite
complex. Indeed, it is not possible to know exactly the time interval
$J'$ associated to $k'$ and therefore we should test all possible
combinations. Another source of complexity is that the token in $E(k',
J')$ can be either in the initial place, $\pp{k'}{J'}$, or in the
place $\pq{k'}{J'}$ (The parameter ${\cal P}$ is used to differentiate
the multiple choices.)

We define the predicate $\NBL(m,t)$ that describes the set of initial
places of the nets $E(k,J)$ such that $k$ is newly-enabled after $t$
fires from $m$.
\[
\displaystyle \NBL(m,t) = \left \{ \pp{k}{J} \divides k \in \nenabl(m,t) \wedge \left ( m
\trans{t} m' \right ) \wedge \left ( J = \SIF(k,m')
\right ) \right \}
\]
Likewise, we define the predicate $\CFL(t,m)$ that describes sets of
places in the net $E(k,J)$ such that $k$ conflicts with $t$ at marking
$m$. The definition of $\CFL$ relies on the relation $t \cfl_m k$,
meaning that $k$ and $t$ are in conflict in the marking $m$, that is
$m \not\geq \Pre(t) + \Pre(k)$. A set of places ${\cal P}$ is in
$\CFL(t,m)$ if it has exactly one place in each transition in conflict
with $t$.
\[
\begin{array}[c]{rcl}
\cfl(m,t) &=& \{  k \divides m \not\geq \Pre(t) + \Pre(k) \}\\[1em]
\CFL(m,t) &=&
\begin{array}[t]{cl}
  \{  \{ r_{k_1,J_1} , \dots , r_{k_n,J_n} \} 
  \divides & r \in \{p, q\} \ \wedge\  \cfl(m,t) = \{ k_1, \dots, k_n \}\\
  & \wedge \  J_1 \in \Itrv(k_1) \wedge \dots \wedge J_n \in \Itrv(k_n) \}\\
  \end{array}
\end{array}
\]
There are at most $2^{|T|}$ sets in $\CFL(m,t)$. Next, we use all
these predicates to formally define the net $E(t,I)$ and, ultimately,
the \TPN $N^\times$.

\begin{definition}\label{def:appa1}
  The net $E(t,I)$ is the 1-safe \TPN such that:
  \begin{itemize}
  \item the set of places is $P^\times$;
    \[
    P^\times = P_{\Marks} \cup \{ \pp{t}{I}, \pq{t}{I} \divides t \in T, I
    \in \Itrv(t)  \}
    \]
  \item the set of transitions is $\{ t_I \} \cup \{ t_{I,m,{\cal P}}
    \divides t \in {\cal E}(m), {\cal P} \in \CFL(m,t) \}$;
  \item the pre- and postconditions of $t_I$ are as in
    Fig.~\ref{fig::appa};
  \item the places in the precondition of $t_{I,m,{\cal P}}$ are $\{
    \pq{t}{I} , m \} \cup {\cal P}$
  \item the places in the postcondition of $t_{I,m,{\cal P}}$ are $\{
    m' \} \cup \NBL(m,t)$
  \item the static time interval of $t_I$ is $I$ and of the
    transitions $t_{I,m,{\cal P}}$ is $[0,0]$;
  \item there is no token in the net in the initial marking;
\end{itemize}
\end{definition}

All the nets $E(t,I)$ have the same set of places. The net $N^\times$
is the 1-safe \TPN obtained by the ``union'' of the nets $E(t,I)$;
places with the same identifier are fusioned and transitions are not
composed.

\begin{definition}
  The net $N^\times$ is the 1-safe \TPN such that:
  \begin{itemize}
  \item the set of places is $P^\times$, as in Def.~\ref{def:appa1};
  \item the set of transitions is $T^\times$;
    \[ 
    T^\times = \{ t_I , t_{I,m,{\cal
        P}} \divides m \in \Marks, t \in {\cal E}(m), I \in \Itrv(t),
    {\cal P} \in \CFL(m,t) \} \]

  \item the static time interval and the pre- and postconditions of
    the transitions in $T^\times$ are as in Def.~\ref{def:appa1};
  \item in the initial marking there is one token in each place
    $\pp{t}{I}$ such that $t$ is enabled at $m_0$, the initial marking
    of $N$, and $I = \SIF(t,m_0)$;
\end{itemize}  
\end{definition}

Our encoding of $N$ is not very concise. Indeed, the best bounds for
the size of $N^\times$ are in $O(|T|.|\Marks|)$ for the number of
places and in $O(2^{|T|}.|T|.|\Marks|^2)$ for the number of
transitions. We can strengthen these bounds if $N$ is a \TPN, that is
when there is only one possible time interval for each transition. In
this case the bound for the number of places is $O(|T| + |\Marks|)$
and the bound for the number of transitions is in $O(|T|.2^{|T|})$.
We can also choose a more concise representation for the markings;
such that we use a vector of places to encode the possible markings
(in binary format) instead of using one place for every single marking
(a representation in unary format).


\subsection{Correctness of our Encoding}
\label{sec:equiv-betw-state}

We start by recalling the notion of (weak) timed similarity between
Timed Transition Systems (TTS) (see for example [The Expressive Power
of Time Petri Nets, B{\'e}rard et al, 2012]). We consider a
distinguished set of actions that stands for ``silent/unobservable
events''; we assume that every silent action as the label $\tau$.  The
weak transition relation $\wtrans{\alpha}$ is defined as $\left (
  \trans{\tau} \right )^* \trans{\alpha}$ if $\alpha \neq \tau$ and as
$\left ( \trans{\tau} \right )^*$ otherwise. Hence we always have $s
\wtrans{\tau} s$ for every state $s$.

\begin{definition}\label{def:bisim}
  Assume $SG_1 = \langle S_1,S^1_0,\rightarrow_1\rangle$ and $SG _2 =
  \langle S_2,S^2_0,\rightarrow_2\rangle$ are two TTS. A binary
  relation $\RR$ over $S_1 \times S_2$ is a \emph{weak timed
    simulation} if, whenever $s_1 \trans[1]{\alpha} s'_1$ in $SG_1$
  then for every state $s_2 \in S_2$ such that $s_1 \RR s_2$ there is
  a state $s'_2$ such that $s_2 \wtrans[2]{\alpha} s'_2$ and $s'_1 \RR
  s'_2$.

  We say that two TTS are (weakly timed) bisimilar, denoted $SG_1
  \approx SG_2$, if there is a binary relation $\RR$ over $S_1 \times
  S_2$ such that both $\RR$ and $\RR^{-1}$ are weak timed simulations.
\end{definition}
 
Next we show that $SG$, the state graph of $N$, and $SG^\times$, the
state graph of $N^\times$, are bisimilar. The definition of $\approx$
depends implicitly on the definition of the silent events $\tau$. In
our case, the only silent actions correspond to the discrete events
$t_{I}$ in $T^\times$. Intuitively, an action of the form $t_{I}$ only
indicates that the transition $t \in T$ has reached its firing
date. It has no effect on the marking (places in $P_{\Marks}$) or on
the other nets $E(k, J)$. On the opposite, an action $t_{I,m,{\cal
    P}}$ commits the decision to fire $t$. We use the action $t$ to
refer to any transition of the kind $t_{I,m,{\cal P}}$ in $SG^\times$.

We list a sequence of properties on the semantics of $N^\times$. Since
this a 1-safe net, we say that a place $r$ is marked on a state
$(m^\times, \tI^\times)$ of $N^\times$ if $m^\times(r) = 1$. The
following properties hold on every reachable state in $SG^\times$:
\begin{enumerate}
\item \label{itm:1st} there is only one token marked in the places $P_{\Marks}$;
\item \label{itm:2nd} if $m$ is marked then there is only one token
  among the collection of (sub)nets $E(t,I)$, for every $t \in {\cal
    E}(m)$. Moreover there are no token in the net $E(k,J)$ if $k
  \notin {\cal E}(m)$;
\item \label{itm:3rd} if the places $m$ and $\pq{t}{I}$ of $E(t,I)$
  are marked then there is exactly one set ${\cal P}$ in $\CFL(m,t)$
  such that $t_{I,m,{\cal P}}$ is enabled; this is the only transition
  enabled in $E(t,I)$.
\end{enumerate}
We can prove these properties by induction on the sequence of
transitions (the path) from the initial state of $SG^\times$ to a
state. If the net $E(t,I)$ is marked, it means that the timing
constraints of $t$, at the time it was newly enabled, was $I$.\\

We define an interpretation function $\interp{.}$ between states of
$SG^\times$ and states of $SG$. Assume $s_\times = (m^\times,
\tI^\times)$ is a state in $SG^\times$, then $\interp{s_\times}$ is
the state $(m, \tI)$ such that:
\begin{itemize}
\item the marking $m$ corresponds to the only place of the kind $m \in
  P_{\Marks}$ that is marked in $m^\times$ (see property~\ref{itm:1st}
  above);
\item for every $t \in {\cal E}(m)$ there is a unique net $E(t,J)$
  marked in $N^\times$ (see property~\ref{itm:2nd} above), then if the
  place $\pp{t}{J}$ is marked we have $\tI(t) = \tI^\times(t_{t,J})$
  and if $\pq{t}{J}$ is marked we have $\tI(t) = 0$.
\end{itemize}
We observe that, with our interpretation, the initial states of
$SG^\times$ is mapped to the initial state of $SG$. Again, using an
induction on the paths of $SG^\times$, it is possible to prove that
every state of the form $\interp{s_\times}$ is reachable in
$N$. Conversely, we prove that every state in $SG$ has a counterpart
in $SG^\times$. Actually, we prove a stronger property that will be
useful to prove the equivalence between state graphs.

\begin{lemma}\label{lemma:1}
  For every state $s \in SG$ there is a state $s_\times \in SG^\times$
  such that $s = \interp{s_\times}$ and for every action $\alpha \in T
  \cup \pReal$; if $s \trans{\alpha} s'$ then there is a state
  ${s'_\times}$ in $SG^\times$ such that ${s_\times} \wtrans{\alpha}
  {s'_\times}$ and $s'= \interp{s'_\times}$.
\end{lemma}

\begin{proof}[sketch] By induction on the sequence of transitions from
  the initial state $s_0$ of $SG$ to $s$. We already observed that
  $s_0$ is the interpretation of the initial state of
  $SG^\times$. Assume that $s = (m, \tI)$ has a counterpart $s_\times$
  in $SG^\times$ and that $s \trans{\alpha} s'$.

  We first study the case of discrete transitions, that is $\alpha =
  t$ with $t \in {\cal E}(m)$. Assume $E(t,I)$ is the net marked in
  $s_\times$ that corresponds to $t$. Since there is a discrete
  transition from $s$, we have that $\tI(t) = 0$, which means that
  either $tI^\times(t_{I}) = 0$ (the transition can fire in
  $N^\times$) or that $\pq{t}{I}$ is already marked. Then there is a
  unique set ${\cal P}$ such that $t_{I,m,{\cal P}}$ can fire, and it
  can fire immediately. This means that there is a state $s'_\times$
  such that $s_\times \wtrans{t} s'_\times$. By definition of
  $t_{I,m,{\cal P}}$, we can choose the same firing dates for the
  newly enabled transitions in $s'_\times$ than in $s'$, hence $s'=
  \interp{s'_\times}$.

  Assume that $\alpha$ is a continuous action $\theta \in \pReal$. We
  need to prove that we can let the time elapse of $\theta$ in the
  state $SG^\times$. We can assume that $\theta \neq 0$, otherwise
  $s'= s$. Since we can let $\theta$ elapse from $s$ we have that
  $\theta \leq \tI(t)$ for every transition $t \in {\cal E}(m)$. By
  definition of $\interp{.}$ we have that $\tI(t) = \tI^\times(t_{I})$
  for some interval $I \in \Itrv(t)$. Then we also have ${s_\times}
  \wtrans{\theta} {s'_\times}$ and $s'= \interp{s'_\times}$.
\end{proof}

Our candidate relation $\RR$ for the bisimulation is the binary
relation from $SG \times SG^\times$ such that $s \mathrel{\RR}
s^\times$ if and only if $s = \interp{s^\times}$. From our previous
results we already have that $\RR$ is total. Hence we just need to
prove that both $\RR$ and $\RR^{-1}$ are simulations. The property for
$\RR$ is a direct corollary of Lemma~\ref{lemma:1}. For the inverse
relation, we assume that $s_\times \trans{\alpha} s'_\times$ in
$SG^\times$. We have three possible case for the action $\alpha$. The
case $\alpha = \tau$ is trivial since, in this case, we have
$\interp{s_\times} = \interp{s'_\times}$. The cases where $\alpha$ is
a discrete transition $t$ or a continuous transition $\alpha \in
\pReal$ is similar than for Lemma~\ref{lemma:1}. Hence $SG$ and
$SG^\times$ are weakly time bisimilar. QED.

\end{document}